\documentclass[11pt,english]{article}
\usepackage{lmodern}
\usepackage[T1]{fontenc}
\usepackage[latin9]{inputenc}
\usepackage{prettyref}
\usepackage{float}
\usepackage{amsthm}
\usepackage{amsmath}
\usepackage{amssymb}
\usepackage{caption}

\usepackage{url}

\makeatletter

\newcommand{\noun}[1]{\textsc{#1}}
\floatstyle{ruled}
\newfloat{algorithm}{tbp}{loa}
\providecommand{\algorithmname}{Algorithm}
\floatname{algorithm}{\protect\algorithmname}

\numberwithin{equation}{section}
\numberwithin{figure}{section}
\theoremstyle{plain}
\newtheorem{thm}{\protect\theoremname}

\date{}

\usepackage{algpseudocode}
\usepackage{algorithm}

\newcommand{\algrule}[1][.2pt]{\par\vskip.5\baselineskip\hrule height #1\par\vskip.5\baselineskip}

\makeatother

\usepackage{babel}
\providecommand{\theoremname}{Theorem}

\newcommand{\eat}[1]{}

\begin{document}

\title{Distributed Approximation Algorithms for the \\Multiple Knapsack
  Problem}

\author{Ananth Murthy \and Chandan Yeshwanth \and Shrisha Rao}

\maketitle

\begin{abstract}
We consider the distributed version of the Multiple Knapsack Problem
(MKP), where $m$ items are to be distributed amongst $n$ processors,
each with a knapsack.  We propose different distributed approximation
algorithms with a tradeoff between time and message complexities. The
algorithms are based on the greedy approach of assigning the best item
to the knapsack with the largest capacity.  These algorithms obtain a
solution with a bound of $\frac{1}{n+1}$ times the
optimum solution, with either $\mathcal{O}\left(m\log n\right)$ time
and $\mathcal{O}\left(m n\right)$ messages, \emph{or}
$\mathcal{O}\left(m\right)$ time and $\mathcal{O}\left(mn^{2}\right)$
messages.
\end{abstract}


\section{Introduction} \label{introduction}

\nocite{pisinger1995algorithms,fidanova2005heuristics,pisinger1999exact,pisinger2005hard}

The Multiple Knapsack Problem (MKP) is a well known optimization
problem which has been studied
extensively~\cite{MartToth90,dawande2000}.  This problem is NP hard,
and can be solved exactly with dynamic programming.  The standard MKP
is, however, only studied in the centralized settings, and its
analogue for distributed algorithms~\cite{AW2004,Lyn96} is heretofore
unknown.

In the distributed setting where knapsacks are dispersed, computation
is divided among different processors that can only communicated by
message-passing, this problem can be a useful model for certain
problems arising in distributed systems, e.g., a data center where
various jobs of different lengths and priority are delegated to
machines with limited
resources~\cite{li2015,korupolu2009,antoniadis2013}.  Then, the
process of selecting the optimum set of jobs to complete with these
limited resources is equivalent to the MKP.  Here, jobs are equivalent
to items, where the processing time of a job is equivalent to the
weight of the item and the priority of the job is equivalent to the
price of the item. The limited machine resources (like processing time
or processing power) is analogous to the fixed capacity of knapsacks.
The question of distributing load across multiple servers is a matter
of practical interest~\cite{willehadson2010method}.

While there is little pertinent work on distributed algorithms for the
MKP, there is pertinent literature on the centralized problem as
arising in various application domains.  Nogueira \emph{et
  al.}~\cite{Nogueira2012} attempt to schedule real time parallel jobs
which are not all known beforehand.  This work attempts to schedule
all jobs arriving real time, in the most efficient way possible. We
however propose a model for a different scenario, where the server
resources (i.e., time) are fixed and the objective is to complete as
many jobs as possible (with allowance for weights for each job). Most
scheduling algorithms assume that balancing load across a servers in a
distributed network is a good approach to obtain an optimal or close
to optimal schedule. This however is not always the case (especially
in the cases where job distributions are heavy tailed)
\cite{harchol98task}.  This finding supports the choice of MKP as a
model for job scheduling, as solving the MKP is inherently different
from balancing the load across all servers.

Islam \emph{et al.}~\cite{islam2009heuristic} look at scheduling jobs
as a multidimensional knapsack problem, where each dimension is
associated with a resource and where each job with some revenue. This
work follows a divide and conquer approach and tries to combine
individual solutions obtained.  However, this model is suited for a
single processor with multiple resources rather than a model with
multiple processors, which is the problem we attempt to solve.  Another
application of the MKP is the Multiple Subtopic Knapsack problem,
to achieve search result diversification~\cite{Yu2014}.  A part of the
knapsack is allocated for and filled with relevant results while the
remaining capacity is used to show diverse results.  Each subtopic is
treated as one of the multiple knapsacks in the standard MKP.

There exist several approximation algorithms based on dynamic
programming after rounding, integer linear programming (ILP) or
various greedy approaches to solve this problem in polynomial time
which obtain solutions within a certain bound of the optimum
solution~\cite{Chekuri00}.  However, LP/ILP approaches to solve the
problem are not apt in a distributed system, as they lead to
non-polynomial message complexity.  Similarly, Bersekas~\cite{Bert82}
has proposed an algorithm for dynamic programming on a distributed
system, but this method also has exponential time/message complexity
in the worst case.  These methods are not particularly suited for the
distributed setting. This conclusion is echoed again by Paschalidis
\emph{et al.}~\cite{Paschalidis2015} who formulate job scheduling as
a Maximum Weighted Independent Set problem.  They use a relaxed linear
programming approach to solve it.  The solution obtained is close to
optimal; however this method requires a non-trivial number of
iterations to converge.

We hence attempt to develop a distributed approximation algorithm
which achieves a trade-off between optimal performance and time and
message complexity.  We aim for an algorithm that has a message
complexity close to $\mathcal{O}(mn)$, as anything higher is
unacceptable in large networks (large $n$) or with a large number of
items (large $m$).  Thus, we attempt to solve this problem in a
distributed system with a focus on low message complexity.

We consider a generalized MKP in a distributed setting, where $n$
processors $p_{j}$ each own a single knapsack $k_{j}$ and have access
to a common pool of items.  Each knapsack has a capacity $W_{j}$.
$K_{j}$ denotes the set of items assigned to the knapsack $k_{j}$.
There are $m$ items indexed by $i$, each having a fixed weight $w_{i}$
and profit or cost $c_{i}$ associated with it.  The objective here is
then to assign each item uniquely to at most one knapsack in such a
way that the sum of prices of all the items across all knapsacks is
maximized, and the sum of weights of the items assigned to every
knapsack is less than the capacity of that knapsack. Mathematically,
the objective function
\[
C=\sum_{j}\sum_{i}c_{i}x_{ij}
\]
has to be maximized, under the constraints
\[
\forall j,\sum_{i}w_{i}x_{ij}<W_{j}
\]
\[
x_{ij}\leq1
\]
 where
\[
x_{ij}=\begin{cases}
1 & \mbox{if item \ensuremath{i} is assigned to \ensuremath{k_{j}}}\\
0 & \mbox{otherwise}
\end{cases}
\]

It will also be useful to define the profit of a knapsack as the sum
of the profits of all the items assigned to that knapsack, i.e.,
$c\left(K_{j}\right)=\sum_{i}c_{i}x_{ij}$.  Further, we define the
notion of the remaining capacity of a knapsack, $r_{j}$, as the
difference of the capacity of the knapsack and the sum of the weights
of the items assigned to this knapsack, i.e.,
$r_{j}=W_{j}-\sum_{i}w_{i}x_{ij}$.

\subsection{The Model}

There are $n$ processors $p_{j}$, $j=1$ to $n$, fully connected
to each other.  We assume there is a distinguished node $S$ (which
can be thought of as the source of these items or a dispatcher for
jobs within a distributed system), with $\left(c_{i},w_{i}\right)$
for each item.  This node stores for the $i$\textsuperscript{th}
item, the index $j$ of the processor $p_{j}$ to which the item is
assigned or $\bot$ if it is unassigned. We also assume that the node
$S$ has the items sorted in the ratio of $\frac{c_{i}}{w_{i}}$.
This node is also connected to all other processors $p_{j}$. We further
assume that this model is failure free and synchronous.

\subsection{Organization}

Section~\ref{approach1} describes a greedy approach to solving the
MKP, along with the analysis of its performance.  In Section
\ref{approach2} we present and analyze two distributed algorithms
which obtain the same solution to MKP, with differing time and message
complexities.  In Section \ref{future} we discuss the downsides of
adapting optimization and dynamic programming to this setting.  In
Section~\ref{conclusion} we finally present our conclusion and the
scope for future work.

\section{Greedy Approach} \label{approach1}

The obvious greedy approach to solve the centralised single knapsack
problem is to assign the ``best'' item (the item with highest
$\frac{c_{i}}{w_{i}}$ ratio) to the knapsack and repeat till no more
items fit into the knapsack.  In the MKP, this approach would imply
assigning the ``best'' item to any knapsack it fits in. Martello and
Toth~\cite{MartToth90} show that the choice of the knapsack in this
case is irrelevant and leads to the same approximation factor for the
worst case.  However, for simplicity, in case the ``best'' item can be
assigned to multiple different knapsacks, we choose the convention to
assign it to the knapsack with the largest remaining capacity,
$r_{j}$.  This approach is described in Algorithm \ref{algo1}.  In each
round, each processor $p_{j}$ sends its remaining capacity $r_{j}$ to
$S$. $S$ assigns one item to each $p_{j}$ from the sorted list in
order of decreasing capacity.  Each processor $p_{j}$ updates its
capacity after receiving an item.  This is repeated till no items can
be fit into any knapsack.

\begin{algorithm}
\begin{algorithmic}[1] 

\item[]
\State{$ItemList \leftarrow ItemList.$SortDecreasingBy$(\frac{c_{i}}{w_{i}})$ }
\State{$i=0$}
\Comment{$S$ has all items sorted by $\frac{c_{i}}{w_{i}}$}
\State{$\forall j, r_j=W_j$}
\algrule

\item[]
\item[{\bf For the processor} $p_j$ {\bf in each round}:]
\State{Send $\left\langle r_{j}\right\rangle$ to $S$}
\State{Receive $\left\langle c_{i},w_{i}\right\rangle$ from $S$}
\State{$r_{j}\leftarrow r_{j}-w_{i}$}

\item[]
\item[{\bf For the source} $S$ :]

\While{$i \leq$ length$(ItemList)$}
	\State{Receive $\left\langle r_{j}\right\rangle$ from all $p_j$}
	\State{$l = (p_j, r_j)$.SortDecreasingBy($r_j$) }		\Comment{Sort by remaining capacity}

	\For{$p_j$ in $l$}
		\If{$w_i \leq r_j $}
			\State{Send item $ItemList[i]$ to $p_j$}	\Comment{Send next item}
		\Else
			\State{Send $\left\langle \bot \right\rangle$ to $p_j$}
		\EndIf
		\State{$i\leftarrow i + 1$}
	\EndFor
\EndWhile

\end{algorithmic}

\caption{\noun{Simple Greedy \label{alg:Simple-Greedy-Approach}Approach}}
\label{algo1}
\end{algorithm}

Lines 4--6 show the procedure followed by each processor $p_{j}$.
Each processor sends its remaining capacity to $S$, receives a new
item and updates its capacity accordingly. The source $S$ (lines
7--18) repeats its procedure as long as there are items left unassigned
(line 7). It receives from each processor $p_{j}$ its remaining capacity
$r_{j}$ (line 8) and then sorts them by decreasing order of $r_{j}$
(line 9). For each processor in this sorted list, $S$ sends the next
item if it fits (lines 11--15). 

\begin{thm} \label{thm1}
At each step, Algorithm \ref{algo1} assigns the best available item
$i$, to the knapsack $p_{j}$, currently having the largest capacity. 
\end{thm}

\begin{proof}

The proof is by contradiction. Suppose that this was not the case
and the algorithm assigns an item $i$ to a knapsack $p_{j}$ where
either the knapsack or the item is not the optimal choice (i.e., the
item with the highest ratio of cost to weight and the knapsack with
the largest capacity). This would imply that either 
\[
\exists i_{2}:\dfrac{c_{i_{2}}}{w_{i_{2}}}\geq\dfrac{c_{i}}{w_{i}}\mbox{\,\,\,\textsc{Or}\,\,\,}\exists j_{2}:r_{j_{2}}\geq r_{j}
\]
But this is not possible as both the items and the knapsacks are
considered in the decreasing order of the ratio of cost to weight
or remaining capacity.  Thus, both the item $i$ and the knapsack $p_{j}$
are optimal, at every step of the algorithm. 
\end{proof}

\subsection{Analysis}

Algorithm \ref{algo1} takes $\left\lceil \frac{m}{n}\right\rceil $
rounds where $m$ is the number of items and $n$ is the number of
processors.  This comes from the fact that there are $m$ items in all
and $1$ item is dispatched to each of the $n$ processors in each
round. The number of messages is exactly $2$ for every item assigned
to some knapsack, which means we have at most $2m$ messages. However,
this the fatal flaw of this algorithm is that it performs arbitrarily
bad in the worst case~\cite{MartToth90}, as shown below.

Consider $n$ knapsacks all of capacity $W$.  Consider $2n$ of items,
the first $n$ of which have cost $2$ and weight $1$ and the remaining
$n$ items having cost and weight both equal to $W$.  Using the
previous algorithm, the first set of $n$ items are chosen, whereas the
optimum solution is to pick the second set of $n$ items.  The ratio of
the solution to the optimum is $\frac{2n}{2W}=\frac{2}{W}$, which can
be arbitrarily bad depending on the value of $W$.

There is a simple remedy to this problem~\cite{MartToth90}.
At the end of the previous algorithm for each knapsack, we pick the
best of the following two options:

\begin{itemize}
\item the solution obtained by the previous algorithm
\item the most profitable unassigned item $i$, with maximum $c_{i}$ and
with $w_{i}\leq W_{j}$ 
\end{itemize}
This can be represented as 
\[
\arg\max\left(c\left(K_{j}\right),\max_{i:w_{i}\leq W_{j}and\sum_{j}x_{ij}=0}\left(c_{i}\right)\right)
\]

Martello and Toth~\cite{MartToth90} prove that the centralised version
of Algorithm \ref{algo2} gives a factor $\frac{1}{2}$ approximation
scheme in the case of a single knapsack problem, and a factor
$\frac{1}{n+1}$ approximation scheme in the case of a multiple
knapsack problem.

\begin{algorithm}
\begin{algorithmic}[1] 

\item[]
\State{$ItemList \leftarrow ItemList.$SortDecreasingBy$(\dfrac{c_{i}}{w_{i}})$ }
\State{$i=0$}
\Comment{$S$ has all items sorted by $\frac{c_{i}}{w_{i}}$}
\State{$\forall j, r_j=W_j$}
\algrule

\item[]
\item[{\bf For the processor} $p_j$:]
\State{Send $\left\langle r_{j}\right\rangle$ to $S$}
\State{Receive $\left\langle c_{i},w_{i}\right\rangle$ from $S$}
\State{$r_{j}\leftarrow r_{j}-w_{i}$}

\item[]
\item[{\bf For the source} $S$ :]

\While{$i \leq$ length$(ItemList)$}
	\State{Receive $\left\langle r_{j}\right\rangle$ from all $p_j$}
	\State{$l = (p_j, r_j)$.SortDecreasingBy($r_j$) }		\Comment{Sort by remaining capacity}

	\For{$p_j$ in $l$}
		\If{$w_i \leq r_j $}
			\State{Send item $ItemList[i]$ to $p_j$}	\Comment{Send next item}
		\Else
			\State{Send $\left\langle \bot \right\rangle$ to $p_j$}
		\EndIf
		\State{$i\leftarrow i + 1$}
	\EndFor
\EndWhile

\item[]
\item[{\bf procedure} Final() :]		\Comment{Executed after initial assignment of items}
\For{$j=1$ to $n$}
	\State{Pick  $\max\left(K_{j},\max_{i:w_{i}\leq W_{j}}\left(c_{i}\right)\right)$}
\EndFor

\end{algorithmic}

\caption{\noun{Modified Greedy Approach}}
\label{algo2}
\end{algorithm}

Lines 1--18 of Algorithm \ref{algo2} remain the same as in Algorithm
\ref{algo1}.  The new refinement is implemented as the procedure
labelled Final (shown in lines 19--21).  At the end of the initial
assignment of all items, the maximum of the current contents of the
knapsack $K_{j}$ and the single costliest item which has not been
assigned is picked as the final content of that knapsack. The
correctness properties for Algorithm \ref{algo2} are the similar to
those of Algorithm \ref{algo1}.  Explicitly, we can say that

\begin{thm} \label{thm2}
Algorithm \ref{algo2} assigns the best of either:
\begin{enumerate}
\item the single largest unassigned item; or
\item the set of items obtained by the greedy approach of assigning
the best item to the largest knapsack. 
\end{enumerate}
\end{thm}

\begin{proof}
(1) follows trivially from line 19. (2) is equivalent to the
  correctness property of Algorithm \ref{algo1} and has been proven as
  Theorem~\ref{thm1}. The same proof holds here.
\end{proof}

This performance bound of $\frac{1}{n+1}$ for Algorithm~\ref{algo2} is
for the centralized version of the Algorithm \ref{algo2}.
Algorithm~\ref{algo2} has the same correctness properties as the
centralized algorithm~\cite{MartToth90}.  Thus, the same proof for the
performance bound holds here as well. Both the algorithms presented do
not inherently exploit the distributed nature of the system; they are
very similar to the centralized greedy approach.c In the next section,
we modify the simple greedy algorithm to exploit the distributed
setting.

\section{Distributed Greedy Algorithm} \label{approach2}

The previous algorithms presented did not exploit the distributed
nature of the setting. The source assigned all items to the knapsacks
in which case the method proposed by Chekuri~\cite{Chekuri00} can be used
to obtain better performance factor of $\frac{1}{e-1}$.  However,
these algorithms require the single node $S$ to carry out all the
computation. 

In the following algorithms, we present a method in which the nodes in
the network themselves decide the assignment of items to knapsacks.
The source $S$ does not have to perform any major computation during
the algorithm; it only broadcasts details of items and receives the ID
of the knapsack to which that item is assigned.  This assignment is
decided by the processors achieving consensus on which processor has
the largest capacity left.  This algorithm still follows the greedy
approach for assigning items, which was outlined in the previous
sections.

Each round is split into two phases, the first in which $S$ broadcasts
the details for an item, and the second in which the nodes choose
which knapsack the item is assigned to.  This is repeated for each
item.  The knapsack is, as before, chosen to be the one with the largest
remaining capacity.  Algorithms \ref{algo3} and \ref{algo4} differ only
in the way in which the nodes identify this knapsack.

\begin{algorithm}
\begin{algorithmic}[1] 

\item[]
\State{$ItemList \leftarrow ItemList.$SortDecreasingBy$(\dfrac{c_{i}}{w_{i}})$ }
\State{$i=0$}
\Comment{$S$ has all items sorted by $\frac{c_{i}}{w_{i}}$}
\State{$\forall j, r_j=W_j$}
\algrule

\item[]
\item[{\bf For the source} $S$ :]
\For{item $i=1$ to $m$}
	\State{Broadcast  $\left\langle w_{i}\right\rangle$ to all $p_j$} \Comment{Details of the next item}
	\State{Receive  $\left\langle j \right\rangle$}
	\State{Assign $i$ to $j$}
\EndFor

\item[]
\item[{\bf For the processor} $p_j$:]
\State{Receive  $\left\langle w_{i}\right\rangle$ from $S$}
\If{$r_j \geq w_i $}
		\State{Broadcast $\left\langle j, r_j \right\rangle$ to all $p_j'$}
	\Else
		\State{Broadcast $\left\langle j, \bot \right\rangle$ to all $p_j'$}
	\EndIf
\State{Receive $\left\langle j, r_j \right\rangle$ from all $p_j'$}
\State{ $m =$ arg $\max_{j'} (r_j')$ from $S$}	\Comment{Reach consensus}
\If{$m=j$}
	\State{Send $\left\langle j \right\rangle$ to all $S$}
	\State{$r_{j}\leftarrow r_{j}-w_{i}$} 
\EndIf

\item[]
\item[{\bf procedure} Final() :]		\Comment{Executed after initial assignment of items}
\For{$j=1$ to $n$}
	\State{Pick  $\max\left(K_{j},\max_{i:w_{i}\leq W_{j}}\left(c_{i}\right)\right)$}
\EndFor

\end{algorithmic}

\caption{\noun{Distributed Greedy Approach}}
\label{algo3}
\end{algorithm}

The source simply broadcasts each item and receives the ID of the
processor to which the item is assigned (lines 4--8).  Each processor
$p_{j}$ broadcasts $\left\langle j,r_{j}\right\rangle $ or
$\left\langle j,\bot\right\rangle $ depending on whether $r_{j}\geq
w_{i}$ (lines 10--14).  Each processor then picks the maximum capacity
of all the knapsack capacities received (lines 15--16).  It then checks
if its capacity is the maximum and if so, notifies the source $S$
(line 18) and updates its capacity (line 19).  This process is repeated
for each item.  Finally, the procedure Final (lines 21--23) is called
for each processor.  This is exactly the same as in the previous
algorithm.

\begin{thm} \label{thm3}
Algorithm \ref{algo3} assigns each item $i$ to the largest knapsack
$p_{m}$ in each round. \end{thm}
\begin{proof}

The proof is by contradiction. Assume that the ``best'' item $i$ in
each round is assigned to a non-optimal knapsack $p_{k}$, where
$p_{k}\neq p_{m}$ and $p_{m}$ is the knapsack with the largest
remaining capacity.  However, all processors broadcast their
capacities and each processor picks the maximum from this set.  Since
this is a failure-free model, all processors pick the maximum from the
same set.  Thus, the item $i$ cannot be assigned to anything but
$p_{m}$.

\end{proof}

This algorithm is runs in exactly $m$ rounds, one for each item.  The
number of messages is $n^{2}$ is each round for consensus and $n$ for
the initial broadcast.  Thus, the algorithm requires
$m\left(n+n^{2}\right) = \mathcal{O}\left(mn^{2}\right)$ messages. As
before, this algorithm obtains a solution at least as good as
$\frac{1}{n+1}$ times the optimum.

Algorithm \ref{algo3} takes $\mathcal{O}(m)$ time instead of
$\mathcal{O}(\frac{m}{n})$, and has a high message complexity of
$\mathcal{O}\left(mn^{2}\right)$.  This can be improved at a further
cost to time, using consensus.  Currently consensus is
$\mathcal{O}(1)$ in time in each round~\cite{AW2004}.

In the next algorithm we present a slightly different approach to
identify the knapsack with the largest remaining capacity.  This is
done with $n$ messages in each round. This will however require
$\mathcal{O}\left(\log n\right)$ time in each round.  We try to exploit
the synchronous properties of this setting.  To do this:

\begin{itemize}
\item First create a rooted binary tree by identifying some edges in
  the network as tree edges, either to a parent or a child
\item Use only the tree edges to achieve consensus. Only the root
  needs to know which processor picks the next item.
\end{itemize}

The main algorithm remains the same as before. A tree
is constructed at the start, and only the consensus part changes.
To construct the tree: 

\begin{itemize}
\item $p_{1}$ is chosen as root
\item The children of $p_{j}$ are taken to be $p_{2j}$ and $p_{2j+1}$. 
\item The parent of $p_{j}$ is $p_{\left\lfloor j/2\right\rfloor }$. 
\end{itemize}

To achieve consensus, after receiving item details from $S$, each node
$p_{j}$ will pick the maximum capacity from all the capacities in the
nodes of the subtree rooted at $p_{j}$ itself and send this to its
parent. Finally $p_{1}$ sends the ID of the processor with the largest
remaining capacity to $S$. This is described in Algorithm \ref{algo4}.

\begin{algorithm}
\begin{algorithmic}[1] 

\item[]
\State{$ItemList \leftarrow ItemList.$SortDecreasingBy$(\dfrac{c_{i}}{w_{i}})$ }
\State{$i=0$}
\Comment{$S$ has all items sorted by $\frac{c_{i}}{w_{i}}$}
\State{$\forall j, r_j=W_j$}
\State{$\forall j, parent=left=right=\bot$}
\algrule

\item[]
\item[{\bf Tree construction} for $p_j$ :]
\State{$parent = p_{\left\lfloor j/2\right\rfloor }$}
\State{$left = p_{2j}$}
\State{$right = p_{2j+1}$}

\item[]
\item[{\bf For the source} $S$ :]
\For{item $i=1$ to $m$}
	\State{Broadcast  $\left\langle w_{i}\right\rangle$ to all $p_j$} \Comment{Details of the next item}
	\State{Receive  $\left\langle j \right\rangle$}
	\State{Assign $i$ to $j$}
\EndFor

\item[]
\item[{\bf For the processor} $p_j$ :]
\item[]
\item[{\bf upon} receiving $w_i$ from $S$:]
\State{execute Consensus()}

\item[]
\item[{\bf upon} receiving $i$ from $S$:]		\Comment{Item received}
\State{$r_{j}\leftarrow r_{j}-w_{i}$}

\item[]
\item[{\bf procedure} Consensus() :]			\Comment{To reach consensus on for a particular item}
\For{$k=\log n$ to $1$}
	\If{$k=\log j$}
		\State{Receive $\left\langle id_1, cap_1 \right\rangle$ from  $left$}
		\State{Receive $\left\langle id_2, cap_2 \right\rangle$ from  $right$}
		\State{$\left\langle id, cap \right\rangle = \left\langle arg \max(r_j,cap_1,cap_2), \max(r_j,cap_1,cap_2) \right\rangle$}
		\State{Send $\left\langle id, cap \right\rangle$ to $parent$}
	\EndIf
\EndFor
\If{$j=1$}
	\State{Send $\left\langle id \right\rangle$ to $S$}
\EndIf

\algstore{algo4}
\end{algorithmic}
\caption{\noun{Greedy Approach with Modified Consensus}}
\end{algorithm}

\clearpage

\begin{algorithm}
\ContinuedFloat
\begin{algorithmic}[] 
\algrestore{algo4}

\item[]
\item[{\bf procedure} Final() :]		\Comment{Executed after initial assignment of items}
\For{$j=1$ to $n$}
	\State{Pick  $\max\left(K_{j},\max_{i:w_{i}\leq W_{j}}\left(c_{i}\right)\right)$}
\EndFor

\end{algorithmic}

\caption{\noun{Greedy Approach with Modified Consensus-Contd.}}
\label{algo4}
\end{algorithm}

Each processor identifies tree edges at the start of the procedure
(lines 5--7). As before the source simply broadcasts each item and
receives the ID of the processor to which the item is assigned (lines
8--12). Each processor $p_{j}$, upon receiving $w_{i}$ from $S$,
starts the consensus subroutine (line 13). To achieve consensus (lines
15--25) each processor receives the maximum capacity from its left
and right sub trees (lines 17--18). It then picks the maximum of these
two capacities and its own capacity and sends this to its parent node
(lines 19--20). This repeats for all processors at each of the $\log n$
levels of the binary tree, starting bottom up (line 15--22). Finally,
the root, $p_{1}$ sends to $S$ the ID of the processor with the
largest capacity (lines 23--25). The source $S$ then sends the item
$i$ to this processor, say $p_{j}$. This processor then updates
its capacity accordingly (line 14). This process is repeated for each
of the $m$ items. Finally, the procedure Final (lines 21--23) is called
for each processor. This is exactly the same as in the previous algorithm. 

\begin{thm} \label{thm4}

Algorithm \ref{algo4} assigns each item $i$ to the largest knapsack $p_{m}$
in each round.

\end{thm}

\begin{proof}

The proof is by induction. We will prove that in each round of the
consensus subroutine, each node $p_{j}$, sends the maximum capacity
of all the nodes present in the sub tree rooted at $p_{j}$ to its
parent. The base case is for the nodes at the lowest level which simply
transmit their capacities to their parent nodes. For the induction
step, assume that this property is satisfied at level $k$ of the
binary tree. Then, each node at level $k$ receives the maximum from
its left and right children (for the left and right sub trees). It
then picks the maximum capacity from amongst these and its own capacity
and transmits it to its parent. Thus, the maximum capacity of the
all the nodes in the sub tree rooted at this node is sent to its parent
at the next level. Thus, this property now holds for the next level
as well. Hence, it also holds for the root node $p_{1}$, which transmits
the maximum capacity of all the nodes to the source (as all the nodes
are children of $p_{1}$). This completes the proof for this theorem. 
\end{proof}

\subsection{Analysis}

Algorithm~\ref{algo4} is $\mathcal{O}\left(m\log n\right)$ in
time.  There are $m$ rounds, one for each item, and consensus takes
$\log n$ phases in each round.  The number of messages is now
$\mathcal{O}\left(n\right)$ is each round for consensus (as each node
transmits only a single message to its parent) and the initial
broadcast.  Thus, the algorithm requires $\mathcal{O}\left(mn\right)$
messages.  This solution obtained is at least as good as
$\frac{1}{n+1}$ times the optimum, as before.

It should be noticed here that each node sending its remaining
capacity directly to the root node $p_{1}$ is not as efficient as the
method described in Algorithm~\ref{algo4}.  If one node has to pick
the maximum value of remaining capacity from a list of $n$ elements,
then it would require $\mathcal{O}(n)$ comparisons and
$\mathcal{O}(n)$ time. Our method requires $\mathcal{O}(n)$
comparisons but $\mathcal{O}(\log n)$ time since these comparisons
happen in parallel.

\section{Further Improvements} \label{future}

In this section we look at other methods to improve the performance
with reference to the optimum.

\subsection{Heuristics}

Heuristics involve switching items between knapsacks to fit more items
in.  Items can be switched one for one, one for two or two for one.
We can even consider more cases of switching---three for one, and so
on. If we do this for all possible combinations of items, we will
eventually achieve the optimum. This will however take exponential
time. Thus we have to restrict ourselves to some limit. However no
performance guarantee can be achieved unless all possible switches are
considered.

Other centralized heuristics for MKP are similar, one of which
involves setting up D-sets (Dominating sets) for every element. A
D-set for an item is the set of all items that are dominated by it,
i.e., the set of items which cannot be included in the solution if the
first item is not included in the solution. This is otherwise the set
of items which have a higher weight and lesser cost than this
item. Once, the D-set is found for each item, optimised selection is
used, where an item and its D-set can be eliminated from
consideration, which takes $\mathcal{O}\left(m\right)$time. However,
computing these D-sets is still expensive and the overall time
complexity remains exponential \cite{MartToth90}.

\subsection{Distributed LP Approaches}

A LP problem can be solved on a distributed system in the following
way---the variables whose values are to be found are split across all
nodes \cite{DiagDom90}. In each iteration, only one variable is
updated on one node and all the other variables are kept fixed.  At
the end of the iteration, this value is updated in all nodes. This
means that we have $\mathcal{O}(n)$ messages for each
iteration. Further, we also have $mn$ variables for the LP. We will
therefore have $\mathcal{O}\left(mn^{2}\right)$ message complexity at
the very least assuming one iteration for every variable. This method
also assumes the diagonal dominance condition for the constraints
(which we have not verified for the MKP).  We also do not know of a
good rounding scheme from a LP solution obtained to an ILP solution
required for the MKP, making this approach infeasible.

The MKP can also be posed as a convex optimization problem with linear
constraints can be solved to obtain close to optimum values
\cite{DistOpt10}.  This assumes that the constraints are positive and
the objective function is separable (which is true for the MKP). This
algorithm uses gradient descent, which may not be easily calculable
for the MKP. This algorithm has inner and outer iterations: the inner
iterations apply gradient descent on a given set of parameters, and
these parameters are chosen by binary search by the outer
iterations. The algorithm also calls as a subroutine, the ``gossip''
algorithm to communicate across the network at the end of each inner
iteration. Like before, the gossip subroutine will lead to a high
message complexity, making this approach infeasible.

\subsection{Distributed Dynamic Programming }

Chekuri~\cite{Chekuri00} suggests that MKP can be solved within an
approximation factor of $1-1/e\thickapprox0.63$ for uniform knapsack
capacities and $1/2$ for non-uniform knapsack capacities.  This is a
far better bound than what we have obtained.  This scheme uses a PTAS
(Polynomial Time Approximation Scheme) for solving single knapsack
problems with an approximation factor of $1-\epsilon$ for each
knapsack. The bound of $1/2$ remains irrespective of the order that
knapsacks are considered.

This scheme implies a DP problem for each knapsack, but solving a DP
problem in a distributed setting is not known to be efficient.
Bertsekas~\cite{Bert82} proposes an algorithm that has
exponential-time convergence in bad cases.  Even with constant message
passing per round, this would still have exponential message
complexity in the worst case.

\section{Conclusion} \label{conclusion}

We have presented distributed approximation algorithms for the MKP,
the best of which has a message complexity of
$\mathcal{O}\left(mn\right)$, time complexity of
$\mathcal{O}\left(m\log n\right)$, and a performance bound of
$\frac{1}{n+1}$. The currently existing methods to obtain better
performance cannot be feasibly implemented on a distributed system
with low message/time complexity (in $\mathcal{O}(n)$ or
$\mathcal{O}(n\log n)$.)

We believe that the MKP can be used as an alternative approach to
scheduling and allocation in distributed systems such as data centers
used in cloud computing.  Our focus on a low message complexity is of
particular importance when the number of items or jobs to be assigned
is very high, as in the case of modern web servers.  A low message
complexity is also necessary when the number of processors is high, as
in a large data center.

The MKP also has applications in other systems such as allocation of
spectra in radio networks~\cite{song2008}, so it stands to reason that
distributed versions of the same would also be of much interest for
similar reasons.

\eat{

A suitable heuristic can be determined to try and exchange items
between knapsacks for a suitable gain in the objective function. The
exact trade-off between message/time complexity and performance can
also be investigated.  Possibly better methods for solving ILP/LP or DP
in a distributed setting can be developed and applied to adapt a good
centralized algorithm for MKP.

Failure/fault tolerant and asynchronous models can also be considered,
which may be more relevant in the real world or for applications of
this problem.

}

\bibliographystyle{plain}
\bibliography{dist-MKP}

\end{document}